\pdfoutput=1
\documentclass[11pt]{article}
\PassOptionsToPackage{hyphens}{url}
\usepackage[a4paper,margin=1in]{geometry}
\usepackage{array}
\usepackage{microtype}
\usepackage[utf8]{inputenc}
\usepackage[hyperfootnotes=false]{hyperref}
\hypersetup{breaklinks=true}
\AtBeginDocument{\Urlmuskip=0mu plus 1mu}
\usepackage[stable]{footmisc}
\usepackage{amsmath,amssymb}
\usepackage{amsthm}
\usepackage{graphicx}
\usepackage[hyphens]{url}
\usepackage{enumitem}
\usepackage{booktabs}
\usepackage[section]{placeins}
\usepackage{array,tabularx,ragged2e,makecell}
\usepackage[font=small,labelfont=bf]{caption}
\usepackage[T1]{fontenc}
\usepackage{lmodern}
\usepackage{stfloats}
\usepackage{titling}
\usepackage{listings}
\setlist{nosep}
\makeatletter\@ifundefined{bibsep}{\newlength{\bibsep}}{}\makeatother

\newtheorem{observation}{Observation}

\title{Bitcoin as an Interplanetary Monetary Standard with Proof-of-Transit Timestamping}
\author{
\begin{tabular}[t]{@{}c@{}}
Jose E. Puente\\
\textit{Mobiletalk-Q, SL}\\
\texttt{jepuente@talk-q.com}
\end{tabular}
\and
\begin{tabular}[t]{@{}c@{}}
Carlos Puente\\
\textit{Mobiletalk-Q, SL}\\
\texttt{carlospuenteg@talk-q.com}
\end{tabular}
}

\setlist{itemsep=0.3em, topsep=0.4em, parsep=0em, partopsep=0em}

\newenvironment{hexlisting}{
  \begingroup\ttfamily\scriptsize
  \begin{tabular}{@{}p{0.76\linewidth}>{\raggedright\arraybackslash}p{0.24\linewidth}@{}}
}{
  \end{tabular}\par\endgroup
}

\begin{document}
\date{}
\maketitle
\begin{abstract}
We explore the feasibility of deploying Bitcoin as the shared monetary standard between Earth and Mars, accounting for physical constraints of interplanetary communication. We introduce a novel primitive, Proof-of-Transit Timestamping (PoTT), to provide cryptographic, tamper-evident audit trails for Bitcoin data across high-latency, intermittently-connected links. Leveraging Delay/Disruption-Tolerant Networking (DTN) and optical low-Earth-orbit (LEO) mesh constellations, we propose an architecture for header-first replication, long-horizon Lightning channels with planetary watchtowers, and secure settlement through federated sidechains or blind-merge-mined (BMM) commit chains~\cite{bip300,bip301}. We formalize PoTT, analyze its security model, and show how it measurably improves reliability and accountability without altering Bitcoin consensus or its monetary base. Near-term deployments favor strong federations for local settlement; longer-term, blind‑merge‑mined commit chains (if adopted) provide an alternative. The Earth L1 monetary base remains unchanged, while Mars can operate a pegged commit chain or strong federation with 1:1 pegged assets for local block production. For transparency, if both time‑beacon regimes are simultaneously compromised, PoTT‑M2 (and PoTT generally) reduces to administrative assertions rather than cryptographic time‑anchoring.
\end{abstract}


\section{Introduction}\label{sec:intro}
Bitcoin's scarcity, decentralized verification, and predictable issuance make it a natural candidate for a universal monetary standard. As humanity moves toward establishing permanent settlements on Mars, a critical question arises: how can Bitcoin function reliably across interplanetary distances? The one-way light time (OWLT) between Earth--Mars ranges from about 3 to 22 minutes over a synodic cycle (from JPL Horizons/NAIF SPICE ephemerides~\cite{naif_horizons}); combined with intermittent contacts and blackouts, this renders synchronized mining impractical, yet leaves room for verification, local payments, and asynchronous settlement \cite{bitcoin_whitepaper}.

To make notation unambiguous for readers and reviewers, Table~\ref{tab:symbols} summarizes the symbols and parameters used throughout the paper (e.g., OWLT, RTT, jitter allowance $J$, and the incremental Lightning margin $\Delta^{\mathrm{extra}}_{\mathrm{CLTV}}$). We operationalize these quantities in Fig.~\ref{fig:cltv_margin_vs_owlt}, which converts Earth--Mars light-time into concrete \texttt{CLTV} blocks, and show the precise layer attachment points for PoTT metadata in Fig.~\ref{fig:interplanetary_bitcoin_stack_integration_map}.

\paragraph{Scope and regime.} While our design targets AU-scale interplanetary distances
typical of a star's circumstellar habitable zone (CHZ), not just Earth--Mars, we use the
Earth--Mars pair as the canonical, near-term case study because it offers clear operational
constraints, traffic models, and measurements. The mechanisms we propose (e.g., PoTT receipts,
latency-aware Lightning policies, and header-first replication) generalize to other CHZ pairs.

Our main contributions are summarized in Section~\ref{sec:contributions}; we review the foundations and prior systems in Section~\ref{sec:soa}. Early relativistic analyses anticipated that global block intervals would need to scale with interplanetary distances to preserve fairness (e.g., \cite{ladha2016relativity}); we instead preserve Bitcoin's base-layer parameters and shift adaptation to higher layers plus PoTT-based transport accountability.

\paragraph{Paper organization.} Section~\ref{sec:contributions} summarizes our main contributions.
Section~\ref{sec:soa} reviews the state of the art across DTN/BPv7/BPSec, lunar networking frameworks
(LunaNet/Moonlight), and Lightning timing/security results relevant to our setting. Section~\ref{sec:sysmodel}
formalizes the system model and assumptions for AU-scale, intermittently connected links in the circumstellar
habitable zone (CHZ). Section~\ref{sec:pott} defines PoTT and its verification rules. Section~\ref{sec:architecture}
presents the end-to-end design, header-first replication, latency-aware Lightning policy, and settlement rails, together
with the CLTV/CSV parameterization. (CSV can be time-based via BIP-68 relative-locktime units with 512-second granularity; height-based CSV should use the general formula with \(b_{\text{target}}\), not a hard-coded ``/10''.) Section~\ref{sec:security} analyzes security and outlines verification profiles.
Section~\ref{sec:ops} details operational guidance and a phased deployment roadmap. Section~\ref{sec:conclusion} concludes. Appendix~A specifies a minimal PoTT wire format and includes a test vector.

\paragraph{Requirement levels.}
We use the keywords MUST, SHOULD, and MAY as defined in RFC~2119 and RFC~8174~\cite{rfc2119,rfc8174} when, and only when, they appear in all capitals.
\section{Contributions}\label{sec:contributions}
We distill the paper's contributions as follows:
\begin{itemize}[leftmargin=*]
  \item \textbf{Interplanetary Bitcoin architecture.} A physics-aware architecture that
  preserves Bitcoin's base-layer parameters while enabling reliable operation across AU-scale,
  intermittently connected links via header-first replication, latency-aware Lightning, and
  asynchronous settlement rails.
  \item \textbf{Proof-of-Transit Timestamping (PoTT).} A new transport-level receipt primitive
  that cryptographically chains hop-timed custody attestations to Bitcoin payload hashes,
  yielding tamper-evident propagation histories suitable for off-chain disputes and operational
  accountability.
  \item \textbf{Latency-aware Lightning policy.} A closed-form parameterization of interplanetary
  timelocks (CLTV/CSV) that incorporates one-way light time and jitter allowances, and a
  packaging of PoTT evidence for watchtowers and adjudicators.
  \item \textbf{Security, verification profiles, and operations.} A threat model and verification
  profiles (e.g., PoTT-M2) that combine time-beacon audits, OWLT envelopes, and administrative
  diversity, plus practical guidance on key management, retention, and phased deployment.
\end{itemize}

\begin{figure}[t]
    \centering
    \includegraphics[width=0.90\textwidth]{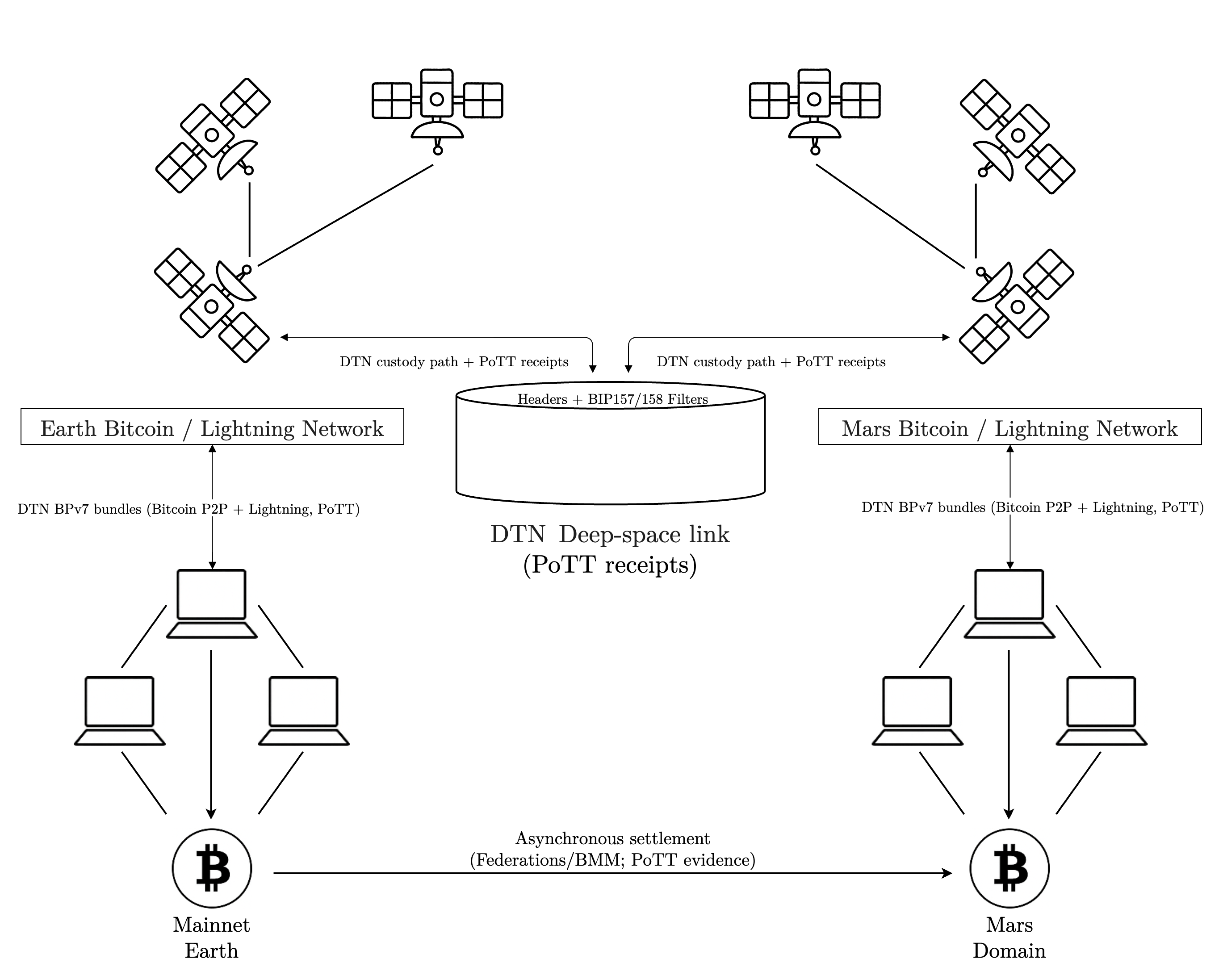}
    \caption{Conceptual architecture for interplanetary Bitcoin using Proof-of-Transit Timestamping (PoTT). DTN deep-space link carries BPv7 bundles with PoTT receipts; each planetary domain runs its own Bitcoin/Lightning network; asynchronous settlement via federations or BMM sidechains.}
    \label{fig:interplanetary_bitcoin_pott_architecture}
\end{figure}

\begin{table}[htbp]
  \centering
  \footnotesize
  \caption{Symbols and parameters used for latency-aware Lightning timelocks and PoTT verification.}
  \label{tab:symbols}
  \setlength{\tabcolsep}{6pt}
  \renewcommand{\arraystretch}{1.15}
  \begin{tabularx}{\linewidth}{@{}>{\RaggedRight\arraybackslash}p{0.22\linewidth}>{\RaggedRight\arraybackslash}X
                                   >{\RaggedRight\arraybackslash}p{0.20\linewidth}@{}}
    \toprule
    \textbf{Symbol} & \textbf{Definition} & \textbf{Units/Example} \\
    \midrule
    \(\mathrm{OWLT}\) & One-way light time Earth--Mars & minutes (3--22) \\
    \(\mathrm{RTT}\) & Round-trip light time; \(\mathrm{RTT}=2\cdot\mathrm{OWLT}\) & minutes \\
    \(\mathrm{RTT}_{\max}\) & Maximum expected RTT over the contact plan & minutes \\
    \(\mathrm{CLTV}\) & Absolute timelock used in HTLCs (BIP65~\texttt{OP\_CHECKLOCKTIMEVERIFY}) & blocks \\
    \(\mathrm{CSV}\) & Relative timelock for commitment/penalty flows (BIP112~\texttt{OP\_CHECKSEQUENCEVERIFY}) & blocks \\
    \(J\) & Contact/jitter allowance (DTN custody, scheduling,\\ processing) & minutes (0, 30, 60) \\
    \(b_{\text{target}}\) & Target L1 block interval & minutes (Bitcoin: 10) \\
    \(\Delta^{\mathrm{extra}}_{\mathrm{CLTV}}\) & Additional \texttt{CLTV} blocks due to interplanetary latency; \(\left\lceil\frac{\mathrm{RTT}+J}{b_{\text{target}}}\right\rceil\) & blocks (1--11) \\
    \(B_{\mathrm{base}}\) & Operator base CLTV margin (Earth policy) & blocks (e.g., 144 at Bitcoin’s 10 min target) \\
    \(M_{\mathrm{op}}\) & Operational margin for queueing/retries & blocks (e.g., 1--2) \\
    \(B_{\mathrm{total}}\) & Recommended interplanetary CLTV; $= B_{\mathrm{base}}+\Delta^{\mathrm{extra}}_{\mathrm{CLTV}}+M_{\mathrm{op}}$ & blocks \\
    \(t_{\mathrm{in}}, t_{\mathrm{out}}\)  & Ingress/egress timestamps encoded in TAI; UTC is for display-only conversions; implementations MUST use TAI consistently. & TAI seconds (CCSDS CUC, epoch 1958-01-01) \\
    \(\mathrm{NodeID}\) & Authorized relay identifier (public key; BIP-340 x-only) used in PoTT & 32-byte public key \\
    $\nu$ & Per-message nonce chosen at the origin; unique per payload hash $h$; relays \textbf{MUST} echo $\nu$ unchanged (see §5 and App.~A) & 16 bytes \\
    \(H(P)\) & Payload digest included in receipts; \textbf{Bitcoin-native} where applicable (double-SHA256 for transactions/headers; BIP157/158 filter identifiers), otherwise SHA-256 over the canonical bytes $P$ & bytes (e.g., 32) \\
    $\delta$ & Safety allowance used in the arrived-before-expiry condition (convert to seconds when used in equations; \(\delta \ge J + 2\sigma_t\)). & minutes (convert to seconds in equations; e.g., $\delta = J + 2\sigma_t$) \\
    \(\sigma_t\) & Time-source uncertainty bound (beacons/clock skew) & minutes (convert to seconds when used in equations; \mbox{$\le 1$} min) \\
    \(\Delta_{\mathrm{MTP}}\) & Policy bound on MTP--UTC skew used in disputes & minutes (convert to seconds when used in equations; e.g., $\le 60$) \\
    \(\kappa\) & Optional reorg margin used in the MTP check & blocks (e.g., 0--6) \\
    \(h_{\mathrm{expiry}}\) & HTLC on-chain expiry height used in the arrived-before-expiry test (BOLT~\#2/\#3); absolute block height & blocks \\
    \bottomrule
  \end{tabularx}
\end{table}

\begin{figure}[t]
    \centering
    \includegraphics[width=0.90\textwidth]{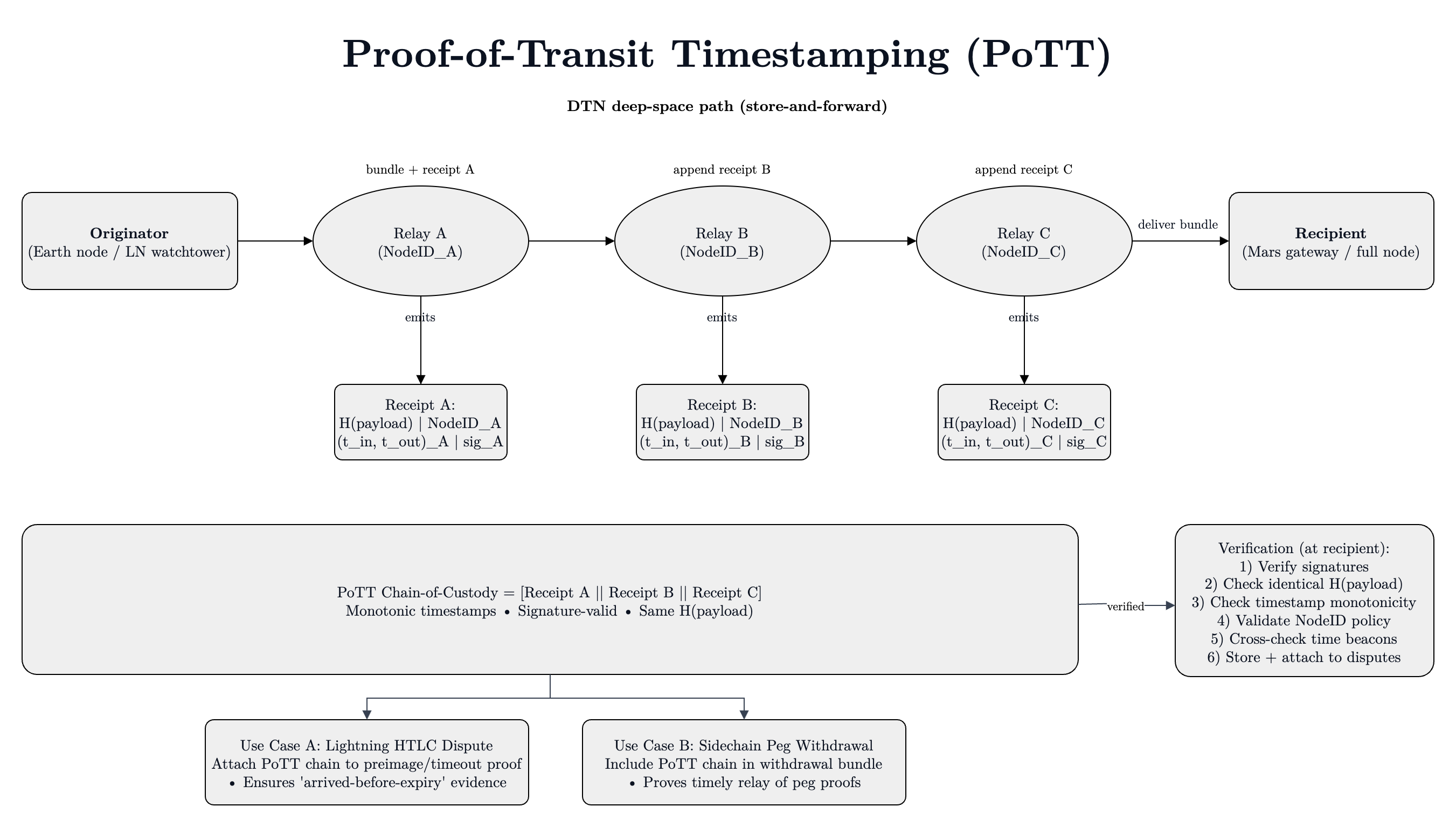}
    \caption{Proof-of-Transit Timestamping (PoTT) detailed flow across a DTN deep-space path. Each relay appends a signed receipt and the concatenated receipts form a verifiable chain-of-custody used for Lightning dispute evidence and sidechain pegs.}
    \label{fig:proof_of_transit_timestamping}
\end{figure}

\begin{figure}[t]
    \centering
    \includegraphics[width=0.88\textwidth]{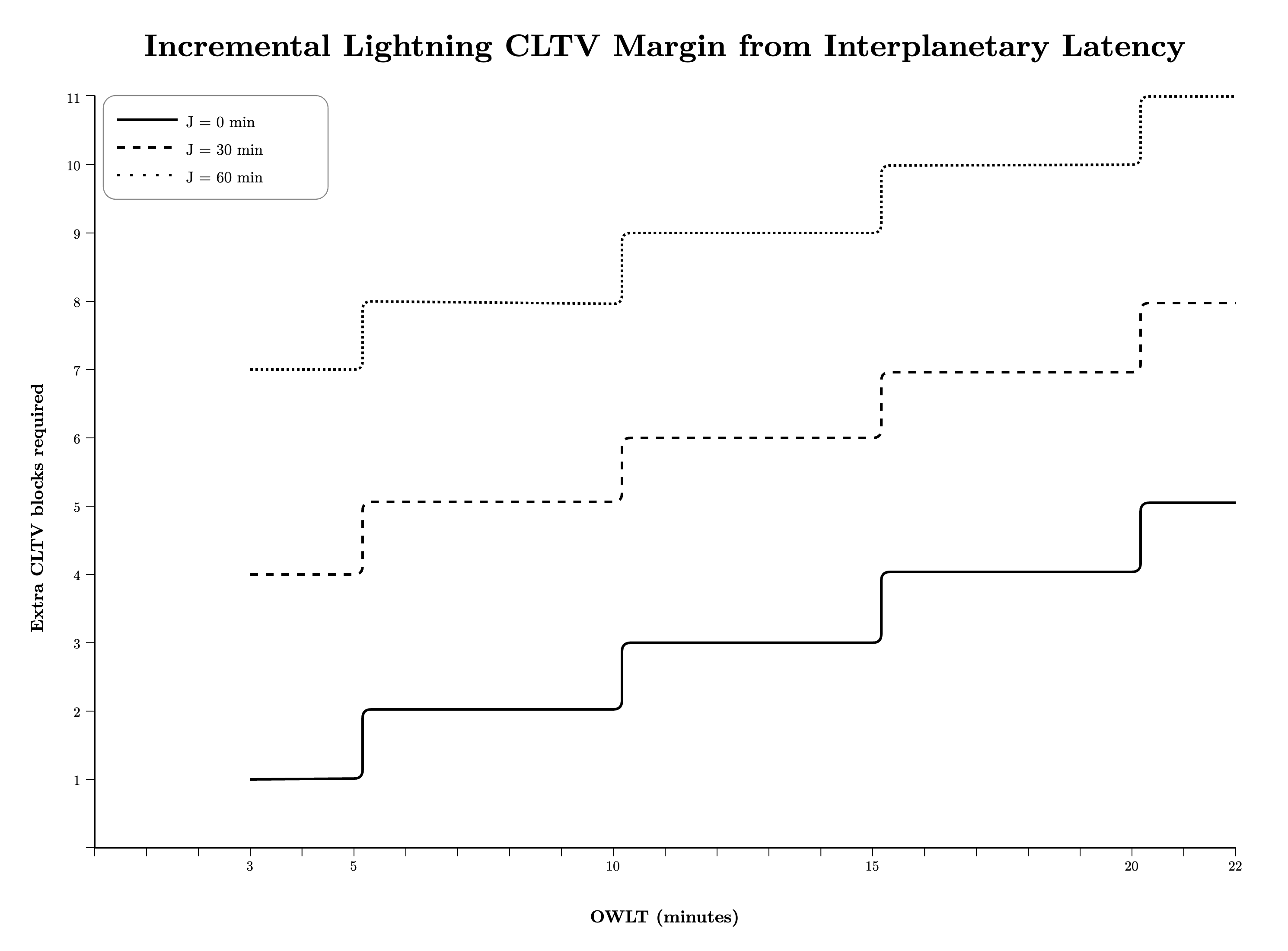}
    \caption{Incremental Lightning CLTV margin due to interplanetary latency. For a given jitter allowance $J$, the additional blocks are $\Delta^{\mathrm{extra}}_{\mathrm{CLTV}}=\left\lceil \frac{\mathrm{RTT}+J}{b_{\text{target}}}\right\rceil$, where $\mathrm{RTT}=2\cdot\mathrm{OWLT}$ and $b_{\text{target}}$ is the L1 target block interval (Bitcoin: $10\,\mathrm{minutes}$). Steps occur at $\mathrm{RTT}+J = k\,b_{\text{target}}$ and increase only when $\mathrm{RTT}+J > k\,b_{\text{target}}$ for $k\in\mathbb{Z}_{\ge 0}$; at equality the value remains $k$. Operators then add their base policy (e.g., 144 blocks).}
    \label{fig:cltv_margin_vs_owlt}
\end{figure}

\begin{figure}[t]
    \centering
    \includegraphics[width=\textwidth]{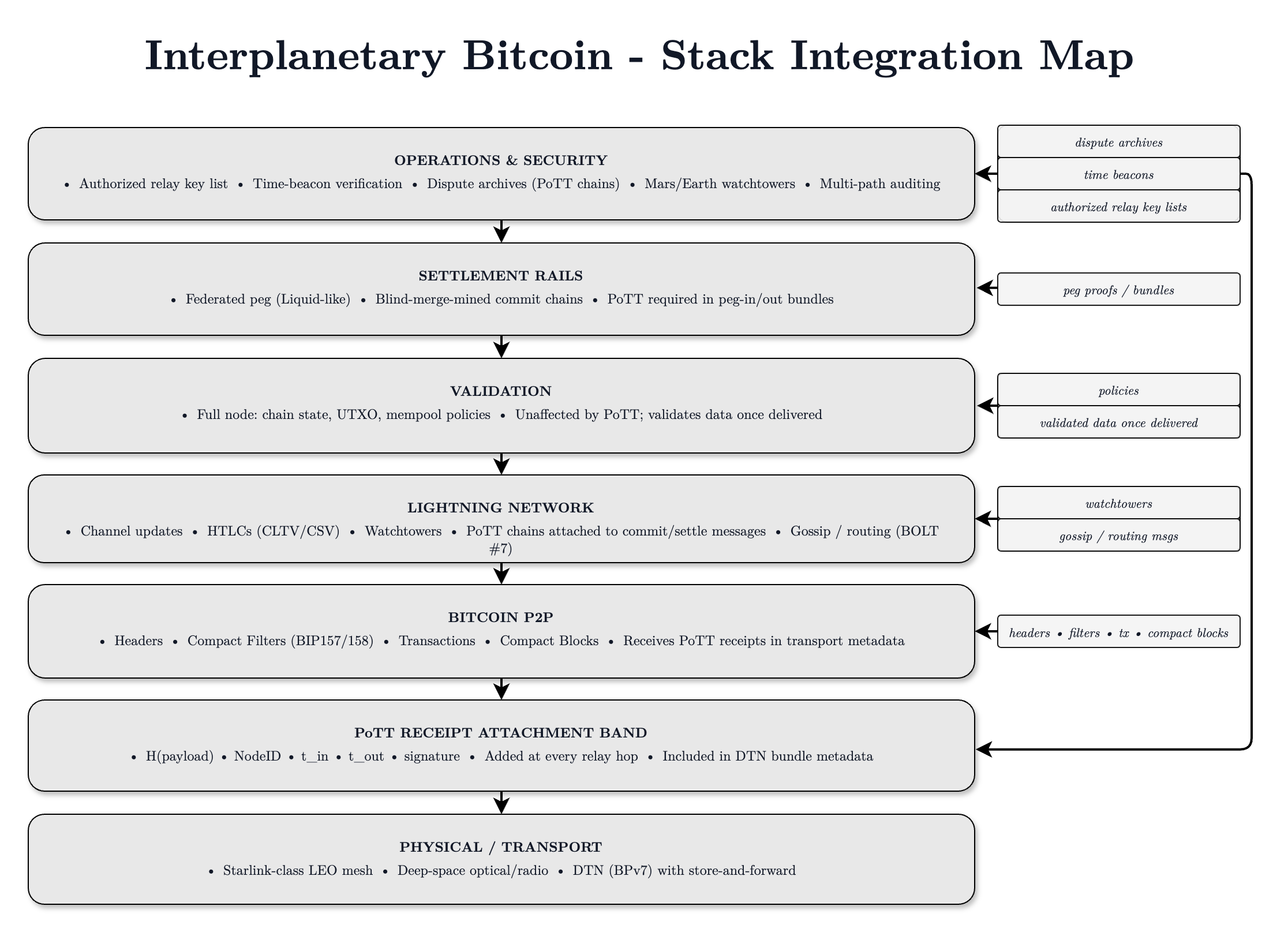}
    \caption{Stack integration map. PoTT receipts are attached in transport metadata (out-of-band) at every relay; Bitcoin P2P carries headers, filters, compact blocks, and transactions. Lightning/watchtowers and settlement rails consume receipts; full-node validation remains unchanged.}
    \label{fig:interplanetary_bitcoin_stack_integration_map}
\end{figure}

\section{State of the Art}\label{sec:soa}
We first review the systems and standards our design builds upon. DTN and the Interplanetary Internet form the transport substrate that makes
planetary-scale networking feasible. The \emph{Bundle Protocol v7} (BPv7) and its
security framework, \emph{BPSec}, specify store-and-forward forwarding, custody
transfer, and hop-by-hop as well as end-to-end protection primitives suitable for
long one-way light times (OWLT) and disconnections~\cite{rfc9171,rfc9172,nasa_dtn_overview}.
Institutional lunar architectures, NASA's \emph{LunaNet} and ESA's \emph{Moonlight}, are
codifying navigation/communications services on top of DTN and are the nearest-term
deployment targets where Bitcoin-related traffic would realistically ride~\cite{lnis_v5_2025,esa_moonlight}.

\textbf{Terminology note.} In the networking literature, “Proof of Transit” usually refers to data-plane path verification/OAM (e.g., SR-based schemes). Our \emph{Proof-of-Transit Timestamping (PoTT)} is different: it produces hop-timed custody \emph{receipts} bound to a Bitcoin payload hash to supply timing evidence and policy signals for Lightning/pegs. PoTT complements BPv7/BPSec; it does not replace them nor assert data-plane path correctness.

\paragraph{Why not just BPSec or Status Reports?}
\emph{Bundle Status Reports (BSRs)} in BPv7 are optional diagnostic messages and are not globally authenticated; they do not create a per-hop custody time chain.
\emph{BPSec} provides end-to-end and/or hop-by-hop authentication of bundle blocks, but it does not yield a signed, per-hop transcript of custody times anchored to external time beacons and an OWLT policy.
PoTT complements BPv7/BPSec by producing a verifiable, beacon-audited, hop-timed chain bound to Bitcoin payload hashes that can be attached to Lightning or peg disputes as off-chain evidence; it does not replace DTN security blocks nor assert data-plane path correctness.

\textbf{Blockchain and payments beyond Earth.} On-Earth, the Bitcoin network has
already been decoupled from the terrestrial Internet via unidirectional broadcast
(Blockstream Satellite) and alternative last-mile relays, demonstrating robustness
to intermittent connectivity~\cite{blockstream_satellite}. In orbit, SpaceChain and
partners have operated blockchain payloads (multisig signing, custody) on the ISS
and commercial satellites, showing that space-based cryptographic operations are
operationally viable even if not consensus-critical~\cite{spacechain_iss}.
These efforts stop short of defining a standards-aligned, physics-aware receipt
layer for cross-planet dispute resolution, which is the focus of PoTT.

\textbf{Lightning and timelocks under long RTT.} Payment-channel security and
capital efficiency under delay have been studied in the context of terrestrial
networks. Sprites introduced constant locktimes across paths and reduced collateral
under adversarial timing~\cite{sprites2017}. Work on routing and reliability
(Pickhardt-Richter) quantifies path selection and failure probabilities relevant
to setting CLTV/CSV policies under uncertainty~\cite{pickhardt2021}. Security
analyses of timing manipulation (``time-dilation''/eclipse) further motivate
auditable transport timing for watchtower policy~\cite{riard2023timedilation}.
The Bitcoin and Lightning standards (BIPs/BOLTs) provide the script/consensus
primitives (CLTV/CSV, compact filters, encrypted transport) PoTT composes with,
without altering L1 consensus~\cite{bip65,bip112,bip157,bip158,bolts2,bolts3,bolts7,bip152,bip324,bip340,bip113}. PoTT is orthogonal to these mechanisms: it strengthens evidence and operational accountability under DTN constraints rather than changing channel-update or settlement protocols.

\textbf{Relativity and macroeconomics across planets.} Prior thought experiments on
\emph{relativistic cryptocurrencies} propose increasing base-layer block intervals
proportionally to network diameter to preserve fairness~\cite{ladha2016relativity}.
Separately, the economics of trade with relativistic delays provides intuition for
interest, discounting, and settlement risk at interplanetary scales~\cite{krugman2010}.
In contrast, our design keeps Bitcoin L1 parameters unchanged and shifts adaptation
to (i) a verifiable transport receipt chain (PoTT) aligned with BPv7/BPSec, and
(ii) L2 policy and dispute mechanisms (Lightning/watchtowers, pegs). To our knowledge,
there is no prior, peer-reviewed architecture that binds \emph{ordered, hop-timed,
cryptographically chained} transit receipts to Bitcoin payload hashes and integrates
those receipts into channel and peg dispute procedures across planetary domains.

\section{System Model}\label{sec:sysmodel}
We consider two primary operational domains, Earth and Mars, with the following elements.
\begin{itemize}
  \item \textbf{Transport:} A combination of optical LEO mesh\,\cite{starlink_isl,starshield} constellations for intra-planet connectivity and deep-space links (e.g., DSOC~\cite{nasa_dsoc}) for interplanetary relays, operated under the Delay/Disruption-Tolerant Networking (DTN) model (Bundle Protocol v7). DTN provides custody transfer, late binding, and contact plans. PoTT receipts are carried in a dedicated DTN extension block parallel to the payload, or as a sidecar bundle, leaving core bundle semantics unchanged.
  \item \textbf{Nodes:} Bitcoin full nodes and Lightning nodes in each domain; optional sidechain federation members; watchtowers for Lightning monitoring.
  \item \textbf{Time sources:} Each relay and gateway maintains a local clock and access to one or more reference time-beacons (e.g., GNSS on Earth, optical two-way time transfer on deep-space links). We denote clock uncertainty bound by $\sigma_t$.
  \item \textbf{Adversary:} Can drop, delay, reorder, or inject messages; can compromise a subset of relays; cannot break standard cryptography (hashes, signatures) nor cause network-wide time-beacons to lie outside $\sigma_t$ for extended periods.
  \item \textbf{Goals:} Preserve Bitcoin's monetary base; enable reliable verification and payments; provide auditable evidence for disputes without modifying Bitcoin consensus.
\end{itemize}

\section{Proposed Approach: Proof-of-Transit Timestamping (PoTT)}\label{sec:pott}
\subsection{Motivation}
In interplanetary operations, disputes frequently hinge on \emph{when} a message arrived relative to a timelock or operational deadline. Missing or delayed headers, transactions, or HTLC updates can force costly rollbacks or channel closures. Existing layers offer no cryptographically verifiable record of message propagation.

\subsection{Formal Definition}
\paragraph{Identifiers and signature scope.}
We define the identifiers carried in PoTT receipts and the exact byte range covered by each relay signature.
\paragraph{Wire format.}

Receipts are serialized as canonical CBOR~\cite{rfc8949} (big-endian integers); the signature covers keys 0--5 as specified. Appendix~A defines the field keys and includes a test vector.
 We standardize that each per-hop signature $s_i$ is a Schnorr signature over the encoded message $M_i = \mathrm{encode}\big(h \parallel \nu \parallel \mathrm{NodeID}_i \parallel t^{(i)}_{\mathrm{in}} \parallel t^{(i)}_{\mathrm{out}} \parallel \mathrm{prev}_i\big)$.
 We instantiate $\mathrm{NodeID}$ as a BIP340 x-only secp256k1 public key (32~bytes) and set $s_i=\mathrm{Sign}_{\mathrm{NodeID}_i}(M_i)$. Other schemes are possible, but we standardize on TAI timestamps; see the ``Timestamp scale'' paragraph.

\paragraph{Timestamp scale.}
PoTT timestamps \textbf{MUST} encode 64-bit International Atomic Time (TAI) seconds using the CCSDS~301 Unsegmented Count (CUC) convention with epoch \textbf{1958-01-01 00:00:00 TAI}. Sub-second fields are \emph{optional}; if present, implementations MAY include a 32-bit fractional field consistent with CUC coarse/fine encoding. This is a wire-format choice only and does not imply the use of UTC on-path. Verifiers \textbf{MUST} use a leap-second table valid at the claimed time $t^{\star}_{\mathrm{TAI}}$ when converting to UTC for display or for comparison with Bitcoin Median-Time-Past (MTP); disputes SHOULD apply a configured bound $\Delta_{\mathrm{MTP}}$ when interpreting ``arrived-before-expiry.'' See CCSDS~301 for details~\cite{ccsds301}. \paragraph{Canonical payload hashing.} \(H\) is the content-identifier used throughout the PoTT chain. When the payload is a Bitcoin object that already has a canonical identifier, \(H(\cdot)\) \textbf{MUST} equal that identifier: double-SHA-256 for block headers and transactions; the BIP157/158 identifiers for compact filters. For all other payload types, \(H(\cdot)\) \textbf{MUST} be SHA-256 of the canonical byte encoding of the payload. The same choice of \(H\) \textbf{MUST} be applied consistently at every hop and is mirrored in Appendix~A's test vector.

Let a payload $P$ (e.g., a header, compact filter, transaction batch, or Lightning update) traverse a path of $N$ relay nodes with identities $\{\mathrm{NodeID}_i\}_{i=1}^N$. Each relay $i$ appends a \emph{receipt}
\begin{equation*}
R_i = (h, \nu, \mathrm{NodeID}_i, t^{(i)}_{\mathrm{in}}, t^{(i)}_{\mathrm{out}}, \mathrm{prev}_i, s_i)
\end{equation*}

\paragraph{Relation to BPSec.}\label{sec:bpsec-relation}
BPSec~\cite{rfc9172} authenticates and/or encrypts bundles at the block level but does not standardize custody timestamps, ordered ingress/egress times, or an anti-splice hash chain bound to a payload digest and a per-message nonce. PoTT is complementary: it yields dispute-grade, monotonic transit evidence without replacing DTN security blocks. In deployments that carry Lightning onion routing, Sphinx packets (BOLT~\#4) are placed in the BPv7 \emph{payload}, while PoTT receipts reside in a BPv7 \emph{extension block}; relays need not understand any Bitcoin semantics to forward bundles.

\paragraph{Chaining and anti-splice.}
Let $h=H(P)$ be the payload hash and let $\nu$ be a 128-bit random nonce minted by the originator for this bundle. Retransmissions of the same payload \textbf{MUST} mint a new $\nu$; verifiers treat each $(h,\nu)$ as a distinct evidence set, and multiple PoTT chains may exist for the same $h$. $\nu$ MUST be unique per payload instance $h$ (not reused across retransmissions) to preclude reuse of identical $(h,\nu)$ pairs across parallel paths. The originator sets $\nu$ once per payload instance and includes it in all bundle copies; relays \textbf{MUST} echo $\nu$ unchanged. 
PoTT binds each hop to the previous one with a hash chain and signs all binding material (specifically, \(h, \nu, \mathrm{NodeID}_i, t^{(i)}_{\mathrm{in}}, t^{(i)}_{\mathrm{out}}, \mathrm{prev}_i\)).

For hop $i\ge 0$ we define
\[
\mathrm{prev}_0:= 0^{256},\qquad 
\mathrm{prev}_i:= H(R_{i-1}\setminus s_{i-1})\ \text{for } i\ge 1.
\]

\noindent\textit{Here $0^{256}$ denotes a 256-bit all-zero string.}
\noindent\textit{Notation:} For any receipt $R_i$, we write $R_i \setminus s_i$ to denote the byte-serialization of $R_i$ with its signature field removed. Consequently, removing or reordering any interior hop changes $\mathrm{prev}_j$ for all subsequent receipts and invalidates their signatures; a chain cannot be cut-and-rejoined without detection.

Each relay $i$ produces a receipt
\[
R_i=\big(h, \nu, \mathrm{NodeID}_i, t^{(i)}_{\mathrm{in}}, t^{(i)}_{\mathrm{out}}, \mathrm{prev}_i, s_i\big),
\]
where the signature covers
\[
M_i = \mathrm{encode}\!\left(h \parallel \nu \parallel \mathrm{NodeID}_i \parallel t^{(i)}_{\mathrm{in}} \parallel t^{(i)}_{\mathrm{out}} \parallel \mathrm{prev}_i\right),
\quad
s_i = \mathrm{Sign}_{\mathrm{NodeID}_i}(M_i).
\]

\textbf{Verification.} A verifier \textbf{MUST} (i) verify all signatures under an authorized relay key list; Operators \textbf{SHOULD} rotate relay signing keys on a fixed cadence and upon suspected compromise, and publish co-signed allowlists and revocation lists via DNSSEC-anchored distribution channels; 
 To limit collusion, verifiers SHOULD require administrative diversity along the custody path (distinct operator keys and disjoint time-beacon sources); chains that fail diversity tests are downgraded in assurance or rejected.
(ii) check $h$ and $\nu$ are identical across all receipts; 
(iii) check timestamp monotonicity $t^{(i)}_{\mathrm{in}}\le t^{(i)}_{\mathrm{out}} < t^{(i+1)}_{\mathrm{in}}$; 
(iv) check $\mathrm{prev}_i=H(R_{i-1}\setminus s_{i-1})$ for every $i\ge 1$; 
(v) apply relay-policy checks and time-beacon audits. 
The inclusion of $(\nu,\mathrm{prev}_i)$ prevents splice/replay attacks in which valid receipts for the same payload hash are recombined into a synthetic path. The detailed end-to-end receipt flow is shown in Fig.~\ref{fig:proof_of_transit_timestamping}. We denote the chain-of-custody as $\mathcal{C}(P) = [R_0 \texttt{||} R_1 \texttt{||} \cdots \texttt{||} R_N]$ carried alongside $P$.

\paragraph{Overhead.} For hash length 32 bytes, key identifier 32 bytes, nonce 16 bytes, ingress/egress timestamps 8 bytes each (TAI seconds; optional 32-bit fractional field per timestamp), previous-link hash 32 bytes, and a 64-byte Schnorr signature, each receipt is $\approx 203$~bytes (roughly 200--205~B depending on CBOR map overhead). Thus chain overhead is $\approx 203N$ bytes for N hops (e.g., 10 hops $\approx$ 2.0~kB). If sub-second fractions are carried, we append a 32-bit fractional field per timestamp (8 bytes total for $t_{\mathrm{in}}$ and $t_{\mathrm{out}}$).

\subsection{Security Properties}
PoTT aims to provide: (i) \textbf{Integrity} of payload and receipts; (ii) \textbf{Authenticity} of relay claims; (iii) \textbf{Ordering/Timeliness} via monotonic timestamps; (iv) \textbf{Publicly attributable claims} at the relay level. PoTT does not prevent censorship, but makes it auditable and attributable to a segment of the path.

\subsection{Integration with Bitcoin Layers}
\begin{itemize}
  \item \textbf{Bitcoin P2P:} PoTT chains are attached in transport metadata (out-of-band) carried alongside headers, filters (BIP157/158), transactions, or compact blocks. Full-node validation remains unchanged; PoTT is not part of consensus. (e.g., BIP152 compact blocks; optional v2 transport BIP324)~\cite{bip152,bip324}.\item \textbf{Lightning:} For cross-planet hashed timelock contracts (HTLCs), \emph{Evidence packaging:} PoTT chains are attached to commit/settle evidence packages for watchtowers and counterparties; we do not alter BOLT~\#2/\#3 wire formats. PoTT chains are included with preimage/timeout evidence and monitored by watchtowers, strengthening disputes about ``arrived-before-expiry.'' The additional \texttt{CLTV} margin is parameterized by Eq.~(\ref{eq:cltv}). PoTT enhances operational accountability and out-of-band arbitration; it does not change on-chain  LN enforceability, which remains purely timelock/script-based. Additionally, operators SHOULD provision CSV deltas for unilateral closes to cover $\mathrm{RTT}_{\max} + J$ plus an operational margin.
  \item \textbf{Sidechains:} Federated pegs or blind-merge-mined commit chains. In our design, we require PoTT evidence for peg-in/out bundles, providing auditability without altering Bitcoin's monetary base.
\end{itemize}

\paragraph{Anchoring ``arrived-before-expiry'' to Bitcoin time.}
When PoTT evidence is used for Lightning HTLC disputes or watchtower policy, timeliness is assessed
against Bitcoin's MedianTimePast (MTP) as defined in BIP-113~\cite{bip113}, not wall-clock time.
Let $C=[R_0,\ldots,R_k]$ be a verified PoTT chain with terminal timestamp $t^{\star}_{\mathrm{TAI}}=t^{(k)}_{\mathrm{out}}$ (TAI). We compare on UTC/Unix seconds by setting $t^{\star}_{\mathrm{UTC}} = t^{\star}_{\mathrm{TAI}} - \Delta_{\mathrm{TAI}\to\mathrm{UTC}}(t^{\star}_{\mathrm{TAI}})$, where $\Delta_{\mathrm{TAI}\to\mathrm{UTC}}(\cdot)$ denotes the current TAI--UTC offset.
Let $H^\star$ be a best-chain header at height $h^\star$ with MTP $T_{\mathrm{MTP}}(H^\star)$.
A verifier \emph{accepts} that a message arrived before expiry if
\begin{equation*}
t^{\star}_{\mathrm{UTC}}+\delta \le T_{\mathrm{MTP}}(H^\star)+\Delta_{\mathrm{MTP}}\quad\text{and}\quad h^\star \le h_{\mathrm{expiry}}-\kappa \,.
\end{equation*}
Here $\delta$ is a safety allowance with $\delta \ge J + 2\sigma_t$, $\Delta_{\mathrm{MTP}}$ is a policy bound on potential MTP--UTC skew (operators SHOULD choose a conservative value; we use $\Delta_{\mathrm{MTP}} = 1\,\mathrm{h}$ in examples). Unless otherwise specified, we set the MTP drift estimation window to \(n = 4032\) blocks (approximately 28 days); dispute bundles \textbf{MUST} include the window and the observed bound, and operators \textbf{MAY} tighten \(n\) with local evidence, and $\kappa$ is an optional reorg margin. Operationally, \(\kappa\) \textbf{SHOULD} exceed an empirical bound on historical MedianTimePast drift estimated over a recent window of \(n\) blocks, in addition to beacon uncertainty, and disputes \textbf{MUST} include the exact block range used for the MTP check. In this paper we set $\delta = J + 2\sigma_t$ by default (operators \textsc{MAY} increase it under degraded beacons). Historic observations of MTP--UTC divergence on mainnet are typically within tens of minutes; therefore we adopt, as a conservative \emph{policy} bound (not an empirical guarantee), $\Delta_{\mathrm{MTP}} \le 1\ \mathrm{hour}$ to accommodate miner timestamp variance while minimizing false accepts. Operators \textit{MAY} tighten this bound with local evidence. \emph{All quantities in the inequality are in seconds; when policy parameters are maintained in minutes (e.g., $J,\sigma_t,\Delta_{\mathrm{MTP}}$), implementations convert by $\times 60$.}
This use of PoTT is strictly out-of-band: it does not alter Bitcoin consensus
or BOLT wire formats; it informs watchtower decisions, routing policy, and off-chain arbitration.
In our prototype, PoTT is carried as a BPv7 extension block within DTN bundles and is consumed by relays and watchtowers; it remains outside Bitcoin consensus and BOLT wire protocols.

\section{Architecture}\label{sec:architecture}
\noindent\textit{Monetary domains.} Earth retains the unchanged Bitcoin L1 as the monetary base. Mars operates a pegged commit chain (we use ``commit chain'' synonymously with a BMM sidechain) or strong federation for local block production with 1:1 pegged assets; no new L1 issuance is introduced. Operationally, we assume strong federations as a practical near-term path and treat BMM commit chains as optional if adopted.

We implement three coordinated layers (see Figs.~\ref{fig:interplanetary_bitcoin_pott_architecture} and \ref{fig:interplanetary_bitcoin_stack_integration_map}).
\begin{enumerate}
  \item \textbf{Header-first replication.} Gateways and relays prioritize \emph{block headers} for timely fork choice and MTP anchoring. At $\sim$52{,}560 blocks/year this is about $80\,\mathrm{B}\times52{,}560 \approx 4.2\,\mathrm{MB/yr}$ (i.e., $\approx 1.07$\,bps). \emph{Compact filters} (BIP157/158) are treated separately: with a conservative median of $20\,\mathrm{kB}$ per block they amount to $\sim 1.05\,\mathrm{GB/yr}$ ($\approx 267$\,bps). Deep-space relays \textit{MAY} therefore ship headers immediately and schedule filters opportunistically (or by checkpoint), preserving safety while matching link budgets; transactions and compact blocks are fetched on demand via the custody path.\item \textbf{Lightning for retail with latency-aware timelocks.} Cross-planet channels and routes set conservative \texttt{CLTV}/CSV deltas. The additional margin due to light-time and jitter is

  \begin{equation}
  \Delta^{\mathrm{extra}}_{\mathrm{CLTV}} = \left\lceil \frac{\mathrm{RTT} + J}{b_{\text{target}}} \right\rceil \text{ blocks} \;\text{(Bitcoin: $b_{\text{target}}=10\,\mathrm{min}$)},
  \label{eq:cltv}
  \end{equation}
\noindent\textit{Here $b_{\text{target}}$ is the target L1 block interval; for Bitcoin $b_{\text{target}}=10\,\mathrm{minutes}$.}
\paragraph{Worked example (sanity check).}
For an Earth--Mars one-way light time $\mathrm{OWLT} = 22$ minutes, the round-trip is $\mathrm{RTT} = 44$ minutes.
With jitter allowance $J=60$ minutes, the additional timelock margin is
$\Delta^{\mathrm{extra}}_{\mathrm{CLTV}} = \left\lceil \dfrac{44 + 60}{10} \right\rceil = 11$ blocks.
With operator base margin $B_{\mathrm{base}} = 144$ and operational margin $M_{\mathrm{op}} = 2$,
the recommended total is $B_{\mathrm{total}} = 144 + 11 + 2 = 157$ blocks.
CSV can be expressed in time using BIP-68 relative-locktime units (512-second granularity), i.e., a time-based CSV of \(t\) seconds corresponds to \(\left\lceil t / 512\ \mathrm{s}\right\rceil\) sequence units~\cite{bip68}.
  
\emph{Note:} On-chain enforcement remains purely script/timelock-based; PoTT evidence informs policy, watchtowers, and off-chain arbitration but does not alter on-chain validity.
\item \textbf{Asynchronous settlement rails.} A strong federation~\cite{liquid_strong_fed} or blind-merge-mined commit chain bridges domains; peg-in/out bundles \emph{must} include PoTT evidence to prove timely relay of peg proofs and to provide tamper-evident logistics histories.
\end{enumerate}

\begin{observation}[Throughput--fairness heuristic bound; proof sketch] Consider two planetary domains with one-way light time (OWLT) $\tau(t)$ between them that varies over a synodic cycle (e.g., Earth--Mars, with $\tau(t) \in [3,22]$ minutes). Any fixed base-layer block interval $b$ that preserves cross-domain mining fairness (bounded stale rate and approximately symmetric tip visibility) across the full cycle must satisfy
\[
 b \gtrsim 2\,\max_t \tau(t) + M.
\]
Here $M$ accounts for validation, queuing, and operational safety margins. In practice this implies $b$ on the order of an hour or longer. Consequently, keeping fairness at base layer via a static $b$ imposes a throughput penalty of at least $b/10$ relative to Bitcoin's current 10-minute cadence.
\end{observation}

\begin{proof}[Sketch]
We model block arrivals as a Poisson process of rate $1/b$ and approximate the worst-case one-way latency by $\max_t \tau(t)$ over a synodic cycle. Let the end-to-end propagation/validation delay be
$D \approx 2\max_t \tau(t)+M$, where $M$ covers validation, queuing, and operational safety margins. In this back-of-the-envelope model, the probability that a found block becomes stale is
$1-e^{-D/b} \approx D/b$ for $D\ll b$. To keep the stale fraction $\le \varepsilon$, it suffices to choose
\[
 b \gtrsim \frac{2\max_t \tau(t)+M}{\varepsilon}.
\]
For Earth--Mars, taking $\max_t \tau(t)\approx 22$ minutes, $\varepsilon=0.05$, and $M$ small relative to $b$ gives $b \gtrsim 44/0.05 = 880$ minutes $\approx 14.7$ hours. This is a heuristic argument (not a formal proof); a careful treatment would require an explicit miner/network model and policy analysis, which we leave for future work. Practically, this motivates keeping L1 parameters fixed and shifting adaptation to higher layers (e.g., PoTT receipts, Lightning/sidechains).
\end{proof}

\section{Security Analysis}\label{sec:security}
\paragraph{Threat model and guarantees.}
\begin{quote}
\textbf{Assumptions and Guarantees.} \emph{Assumes} unbroken hash/signature primitives; availability of signed time beacons; at least one honest relay on each verified path; and integrity of ephemerides/contact plans used for OWLT envelopes. \emph{Guarantees} custody attestation via a per-hop receipt chain, per-hop ordering with monotonic timestamps, and splice-resistance across the path; PoTT \emph{does not} guarantee liveness or delivery.
\end{quote}
\paragraph{Assurance summary.} \emph{PoTT provides accountability---tamper-evident custody plus per-hop timing---but it does not by itself provide liveness or censorship-resistance; those remain policy- and topology-dependent. Our verification profiles (e.g., \textbf{PoTT-M2}) mitigate collusion by requiring administrative diversity and independent time-beacon audits, and can be strengthened with path diversity for high-stakes disputes.}

In the degenerate case where \emph{all} relays on the path collude and public time-beacons are compromised, PoTT evidence reduces to administrative assertions about custody (NodeIDs and path admission); verifiers \emph{SHOULD} treat such chains as non-probative absent corroboration from independent anchors. In keeping with this, PoTT is intended to degrade gracefully to administrative assertions; policy profiles (Table~\ref{tab:pott-m2}) \textbf{SHOULD} require at least two independent, TAI-traceable timing anchors for any path whose evidence is used in adjudication.

 PoTT provides \emph{custody attestation with monotonic timestamps}: it binds an ordered set of relay claims to an unchanged payload and discloses ingress/egress times at each relay. It does \emph{not} guarantee liveness or delivery; relays may drop traffic after signing, and censorship remains possible. Absolute time need not be perfectly synchronized, verification requires monotonic per-hop times within policy bounds, and consistency with external anchors and contact-plan geometry.

\paragraph{Relation to BPSec.} See \S\ref{sec:bpsec-relation}.\paragraph{Time anchors and policy profile.} Verifiers SHOULD check PoTT chains against (i) signed public time-beacons on each domain (e.g., Earth and Mars) and (ii) the ephemeris-derived one-way light time (OWLT) envelope for the claimed window. We recommend a practical profile, \emph{PoTT-M2}: at least three relays from at least two administrative domains, with at least one anchor from each domain; per-hop times must satisfy $|\Delta t|\le J$ and be consistent with the current OWLT envelope. Operators maintain authorized relay key lists and revoke keys that fail audits. By default, for high-stakes disputes, verifiers \textbf{MUST} require at least two path-diverse PoTT chains anchored to disjoint time-beacon regimes; a single chain \textbf{MAY} suffice for routine settlements.

 As a concrete instantiation of these checks, see Table \ref{tab:pott-m2}.
\begin{table}[t]
\centering
\caption{Verification profile \textbf{PoTT-M2}. Operational knobs for verifiers.}

In deployment, a receiver’s watchtower may reject HTLCs that lack an accompanying PoTT chain per a stated policy profile (e.g., PoTT-M2), providing an explicit enforcement hook.
\label{tab:pott-m2}
Operators \textbf{MAY} enforce PoTT-M2 compliance in watchtower policy and federation membership, rejecting non-compliant evidence packages.

\begin{tabular}{p{0.42\linewidth}p{0.5\linewidth}}
\toprule
Parameter & Value/Requirement \\
\midrule
Minimum relays (hops) & $\geq 3$ receipts \\
Administrative diversity & $\geq 2$ distinct operator domains; at least one anchor from each planetary domain \\
Time anchors & Signed public time-beacons on each domain (Earth, Mars) \\
Per-hop timing bound & $|\Delta t| \le J$ and consistent with current OWLT envelope \\
Path diversity (high-stakes) & Require $\geq 2$ path-diverse chains (disjoint operators/time-beacon regimes) \\
Authorized relay keys & Maintain allowlist; revoke failing audits and exclude unauthorized NodeIDs \\
Chain size cap & $\leq 32$ hops \emph{or} $\leq 8$ kB of receipts per bundle \\
Retention window & Keep receipts $\geq 90$ days; longer if attached to disputes \\
\bottomrule
\end{tabular}
\end{table}
\noindent\emph{OWLT envelope source:} Verifiers SHOULD derive the one-way/round-trip light‑time envelope from published contact plans or ephemerides (e.g., SPICE kernels) and record the version used during verification.

\noindent\textit{Assumption (PoTT-M2).} At least one audited relay on the verified path is honest, and the public time-beacon regimes on each planetary domain remain within $\sigma_t$ of true time; if all relays collude and external beacons are compromised, PoTT evidence degrades to administrative assertions.

\paragraph{Multipath receipts and diversity policy.} For high-stakes disputes, verifiers \textbf{SHOULD} require evidence from at least two path-diverse PoTT chains (disjoint relay operators and, when feasible, disjoint time-beacon regimes). When multipath was available but only a single path is presented, adjudicators may treat the evidence as weak.
Verification policy rejects any PoTT chain whose timing is only explainable by back-dating beyond the declared clock-uncertainty bound $\sigma_t$ and the published contact plans; chains must be consistent with beacon cross-checks.

\paragraph{Relay collusion or forged timestamps.} Multiple independent time-beacons and cross-checks on contact plans limit feasible skew. Multi-path routing yields divergent chains if a colluding subset lies about timing.
\paragraph{Receipt omission or truncation.} Verifiers enforce minimum-hop and diversity policies; missing receipts invalidate policy compliance. DTN custody supports custody-based retransmission (RFC 9171) on failure.
\paragraph{Sybil relays.} Operators maintain authorized relay key lists with revocation. Watchtowers and federation members reject chains containing unauthorized NodeIDs.
\paragraph{Denial-of-service via oversized metadata.} Cap PoTT chains to \(\le 32\) hops or \(\le 8\,\mathrm{kB}\); compress receipts; use per-hop admission control.
\paragraph{Privacy.} PoTT receipts expose hop count and coarse path structure to any party that inspects them. Deployments SHOULD minimize leakage by using the privacy mode (commit-and-reveal) described below.
\paragraph{Privacy mode (commit-and-reveal).}
Relays \textbf{MUST} onion-encrypt receipt metadata hop-to-hop (Sphinx per BOLT~\#4~\cite{bolts4}); routine attestations \textbf{SHOULD} redact per-hop times, and only disputes reveal full transcripts. For routine attestations, relays commit to the full receipt transcript by publishing
$H_{\mathrm{txpt}} = H(R_0 \parallel \cdots \parallel R_N)$.
The wire conveys only the tuple $(H_{\mathrm{txpt}}, \min(t_{\mathrm{in}}), \max(t_{\mathrm{out}}), \mathrm{hopcount})$.
Thus observers learn only $[\min(t_{\mathrm{in}}),\max(t_{\mathrm{out}})]$ and hopcount; relay identities and per-hop timestamps remain encrypted and are revealed only if a dispute requires opening the transcript. In case of dispute, parties reveal the transcript and verify (i) that the recomputed $H_{\mathrm{txpt}}$ matches,
(ii) per-hop times lie within $[\min(t_{\mathrm{in}}), \max(t_{\mathrm{out}})]$, and (iii) policy bounds are met.
This achieves minimal disclosure while preserving verifiability of end-to-end timing.
Deployments \textbf{SHOULD} minimize leakage by aggregating to $[\min(t_{\mathrm{in}}),\max(t_{\mathrm{out}})]$ and hop count for routine attestations, carrying receipts under layered onion encryption so each relay learns at most its predecessor and successor, and disclosing full transcripts only during disputes once policy requirements are satisfied (cf.~Table~\ref{tab:pott-m2}).\section{Operational Considerations and Roadmap}\label{sec:ops}
\textbf{Key management.} Relays use hardware-backed keys with periodic rotation; public keys distributed via signed manifests. Relay allowlists are distributed as signed manifests (e.g., DNSSEC-anchored or RFC-style registries) with routine audits and revocation procedures.\newline

\textbf{Cross-domain trust (Earth$\leftrightarrow$Mars).}
Relay allowlists are published as DNSSEC-anchored manifests with mirrored roots in each planetary domain; a federated governance group co-signs periodic checkpoints and CRLs.
Revocations propagate with maximum staleness bounded by the manifest TTL (hours--days); verifiers MUST account for this latency when auditing evidence that spans domains.
\textbf{Time synchronization.} Combine GNSS (Earth), onboard oscillators, and deep-space two-way time transfer (e.g., CCSDS time-code formats and interoperability \cite{ccsds301}); keep $\sigma_t$ within operational bounds.\newline If beacons are unavailable or exceed the configured uncertainty bound $\sigma_t$, operators SHOULD increase the verification margin $\delta$ and/or quarantine PoTT evidence until clocks recover; for high‑stakes operations (e.g., peg withdrawals) systems MUST fail closed.
 Operators \textbf{SHOULD} source at least two independent, TAI-traceable public time beacons (e.g., DSN/USNO and ESA/ESTEC) and record the beacon identities alongside any PoTT evidence attached to disputes.
\textbf{Default margins.} Unless otherwise stated, set $\delta = J + 2\sigma_t$; operators \emph{MAY} increase this during degraded beacons.\newline
\textbf{Retention.} Store PoTT chains for a fixed horizon (e.g., 90 days) unless attached to disputes.\newline
\textbf{Phased deployment.} Phase 0: Earth testbeds; Phase 1: cis-Mars demos; Phase 2: Mars orbit/surface networks; Phase 3: scale security via blind merge-mining and mature governance.

\section{Conclusion}\label{sec:conclusion}
PoTT supplies the missing accountability layer for delay-tolerant Bitcoin without any L1 consensus changes; together with header-first replication and latency-aware Lightning policy, it enables a practical interplanetary Bitcoin economy.
Physics forbids synchronous, cross-planet competitive mining at Bitcoin's current parameters,
but it does \emph{not} forbid a Bitcoin-only economy across planets. This paper showed that
by (i) prioritizing header-first replication for chain awareness, (ii) operating Lightning with
latency-aware CLTV/CSV policies derived from one-way light time (OWLT) and jitter budgets, and
(iii) settling asynchronously via strong federations or blind-merge-mined (BMM) commit chains (we use ``commit chain'' synonymously with BMM sidechain),
a practical and incrementally deployable architecture emerges for interplanetary Bitcoin.

The core technical contribution is \textbf{Proof-of-Transit Timestamping (PoTT)}, a transport-level
receipt layer that binds an ordered, hop-timed chain of custody to the hash of the transported
Bitcoin payload (headers, filters, transactions, compact blocks, or Lightning updates).
PoTT does \emph{not} alter consensus or script semantics; instead, it raises the quality of
\emph{evidence} available to watchtowers, counterparties, and peg operators. By combining
per-hop signatures, a splice-resistant hash chain, and TAI-encoded ingress/egress times, PoTT
turns otherwise opaque DTN store-and-forward behavior into auditable artifacts. We specified a
verification profile (PoTT-M2) that cross-checks receipts against authorized relay keys, public
time-beacons, and OWLT envelopes, and we showed how evidence packages anchor ``arrived-before-expiry''
claims to Bitcoin's MedianTimePast (BIP-113) without introducing new consensus rules. PoTT is strictly \emph{out-of-band}: receipts are consumed by watchtowers and peg federations; no changes to Bitcoin consensus or BOLT wire formats are required.

From an operations perspective, our design keeps failure domains small and controllable. Header-first replication requires only an $\mathcal{O}(\mathrm{MB/yr})$ \emph{annual budget}---about $4.2\,\mathrm{MB/yr}$ for headers (i.e., $\approx 1.07\,\mathrm{bps}$ sustained)\footnote{Computation: headers at $80$\,B per block $\times$ $52{,}560$\,blocks/yr $\approx 4.2$\,MB/yr $\approx 1.07$\,bps; compact filters at $20$\,kB per block $\times$ $52{,}560$\,blocks/yr $\approx 1.05$\,GB/yr $\approx 267$\,bps.}---and compact filters at $20\,\mathrm{kB}$/block require $\approx 1.05\,\mathrm{GB/yr}$ (i.e., $\approx 267\,\mathrm{bps}$). Lightning
channels pay a predictable timelock premium $\Delta^{\mathrm{extra}}_{\mathrm{CLTV}} = \lceil\frac{\mathrm{RTT}+J}{b_{\text{target}}}\rceil$
blocks (Fig.~\ref{fig:cltv_margin_vs_owlt}), and PoTT receipts add kilobytes, not megabytes, per interplanetary
bundle. The architecture composes with existing BIPs/BOLTs (e.g., BIP152, BIP157/158, BIP324) and
DTN/BPv7/BPSec without special treatment from Bitcoin full nodes. This separation of concerns
preserves Bitcoin's monetary base and validation model while letting higher layers internalize
the unavoidable costs of distance.

\paragraph{Limitations and scope.} PoTT provides accountable \emph{custody attestation}, not liveness;
relays may still censor or drop traffic. Our threat model assumes sound cryptography, bounded
clock uncertainty, and independently operated time-beacons; extreme faults in these primitives
can weaken evidence quality. Finally, we have not attempted to equalize miner fairness across
planets at L1, a path that would require hour-scale blocks and collapse global throughput.
Instead, we accept temporarily segmented fee markets.
During periods of opposition or limited interdomain connectivity, participants should expect temporary rebalancing costs and FX-like spreads across domains; these frictions diminish as liquidity and path diversity improve, and rely on asynchronous pegs and periodic
liquidity rebalancing across domains.

\paragraph{Implications.} Under these constraints, Bitcoin can function as a shared monetary
standard between Earth and Mars: verification remains local and cheap; retail payments proceed
over LN with safety margins that scale with physics; and high-value flows settle asynchronously
with tamper-evident logistics histories. In short, PoTT supplies the missing accountability layer
for delay-tolerant Bitcoin, turning a hard physical limit (light time) into explicit policy knobs
instead of consensus changes.

\section{Future Work}
While our proposal is already concrete enough for early prototyping, three areas stand out as the most impactful next steps:

\begin{enumerate}[leftmargin=*]
  \item \textbf{Formal security and timing proofs.} Develop a rigorous, game-based security proof for PoTT's anti-splice construction and a complete end-to-end timing model that incorporates beacon uncertainty, contact-plan geometry, and DTN queueing. This would strengthen confidence in PoTT as a standard.
  \item \textbf{Interoperable specifications and reference implementation.} Produce an open, I-D/RFC-style specification for PoTT's DTN extension block and receipt wire format, and implement a public reference library with Bitcoin/Lightning integration for watchtowers and sidechain pegs.
  \item \textbf{End-to-end testbeds.} Run terrestrial emulation of AU-scale RTT and blackout conditions, followed by public pilots linking separated ground stations with header-first replication, Lightning channels, and PoTT evidence enforcement.
\end{enumerate}

These steps would take PoTT from a research proposal to an operational, cross-vendor standard ready for deployment in interplanetary Bitcoin networks.

\appendix
\section*{Appendix A: Minimal PoTT Wire Format and Test Vector}

\noindent\textbf{Receipt structure.} Receipts use \emph{canonical CBOR} maps with the following integer keys and types (canonical ordering by key):\par
\vspace{0.3\baselineskip}

\noindent\begingroup
\setlength{\emergencystretch}{2em}
\renewcommand{\arraystretch}{1.08}
\begin{tabularx}{\linewidth}{@{}>{\ttfamily}l >{\arraybackslash}X@{}}
0: h     & \textit{bstr} (32) --- payload digest $H(P)$ (Bitcoin-native: double-SHA256 for transactions/headers; BIP157/158 filter identifiers; otherwise SHA-256 of canonical $P$) \\
1: $\nu$ & \textit{bstr} (16) --- per-message nonce; MUST be unique per payload instance $h$ (not reused across retransmissions) \\
2: node  & \textit{bstr} (32) --- relay NodeID (public key; BIP-340 x-only) \\
3: tin   & \textit{int} --- ingress time (TAI seconds; CCSDS CUC, epoch 1958-01-01) \\
4: tout  & \textit{int} --- egress time (TAI seconds; CCSDS CUC, epoch 1958-01-01) \\
5: prev  & \textit{bstr} (32) --- prev$_i = H(R_{i-1}\setminus s_{i-1})$ --- anti-splice binding to the prior hop's receipt sans signature \\
6: sig   & \textit{bstr} (64) --- BIP340\allowbreak{} Schnorr over the canonical CBOR\allowbreak{} of keys 0--5 \\
\end{tabularx}
\endgroup

\medskip
\noindent\textbf{Minimal test vector (normative).} The following CBOR hex encodes one receipt (keys 0--6). Values are provided in full for reproducibility; the signature is a dummy (parser-checkable).

\begin{hexlisting}
  00\ 58\ 20\ 83 a0 12 ac 61 2c 83 f6 89 17 73 87 35 34 65 fb 96 13 56 e8 1b cd 8a da 4b a0 d6 57 da 1c 26 85\ \  & \#0:h = H(P) \\
  01\ 50\ 22 19 c6 46 c0 c3 53 d1 87 ef b2 ca b9 ef 61 5b\ \  & \#1:nu \\
  02\ 58\ 20\ d4 06 3a ea 17 03 81 ce ca f4 d4 3b 1e 8d d3 2e c1 34 9f ac 78 ed c0 75 ce 08 fb 36 4d 60 40 43\ \  & \#2:node \\
  03\ 1B\ 0000000065B9B8A0\ \  & \#3:tin \\
  04\ 1B\ 0000000065B9BD40\ \  & \#4:tout \\
  05\ 58\ 20\ 2c 77 0e 00 80 83 e6 2a fd 13 76 98 ce 19 6d b6 5c b4 06 eb 2b 4c 50 6c b6 fa 0c 54 6f 95 d8 55\ \  & \#5:prev \\
  06\ 58\ 40\ db d5 95 30 45 c5 b1 31 a2 5e ca bd 6f 2d 78 6b 28 7e e1 da 3a e2 84 5b 27 89 b5 1c cd c3 82 ef 83 68 e0 36 50 87 9c 71 75 5b 7f da 46 6b 44 a7 32 18 f6 82 06 25 e9 59 2f cc b3 a6 13 3b 92 b2\ \  & \#6:sig (dummy) \\
\end{hexlisting}

Implementations \textbf{MUST} reject non-canonical encodings and unknown keys.

\section*{Acknowledgments}
We thank the open-source Bitcoin and space networking communities for public documentation and standards that informed this work.

\end{document}